\documentclass{llncs}
\usepackage[utf8]{inputenc}

\usepackage{microtype}
\usepackage{amssymb,amsfonts,amstext,amsmath}

\usepackage{tikz}
\usetikzlibrary{shapes,arrows}
\usetikzlibrary{positioning}
\usetikzlibrary{fit,calc}

\usepackage[onelanguage,ruled,vlined,linesnumbered,norelsize]{algorithm2e}

\usepackage{multirow}
\usepackage{multicol}

\newtheorem{notation}[definition]{Notation}

\date{September, 7th 2015, revised May, 26th 2016.}

\newcommand{\mr}[2]{\multirow{#1}{*}{#2}}
\newcommand{\mc}[3]{\multicolumn{#1}{#2}{#3}}
\newcommand{\mcm}[3]{\multicolumn{#1}{#2}{\ensuremath{#3}}}
\newcommand{\pM}{\phantom{M}}

\newcommand{\sub}[1]{$_{#1}$}

\newcommand{\lift}[1]{\boldsymbol{#1}}

\newcommand{\ssss}[1]{\begin{scriptscriptstyle}#1\end{scriptscriptstyle}}
\newcommand{\sss}[1]{\begin{scriptstyle}#1\end{scriptstyle}}

\newcommand{\coeffss}[1]{\mathsf{\ssss{#1}}}

\newcommand{\FF}{{\mathbb F}}

\newcommand{\qq}{\mathfrak q}
\newcommand{\rrr}{\mathfrak r}

\newcommand{\Z}{{\mathbb Z}}
\newcommand{\ZZ}{{\mathbb Z}}

\newcommand{\QQ}{{\mathbb{Q}}}

\DeclareMathOperator{\GF}{GF}

\newcommand{\norm}[1]{{|\|#1\|}}

\DeclareMathOperator{\Reslt}{Res}
\DeclareMathOperator{\Norm}{Norm}

\DeclareMathOperator{\Prb}{Pr}
\DeclareMathOperator{\LLL}{LLL}

\usepackage{hyperref}
\hypersetup{
   pdftitle={Computing individual discrete logarithms faster with NFS-DL},    
   pdfauthor={Aurore Guillevic},            
   pdfsubject={Number Field Sieve},           
   colorlinks=true,         
   linkcolor=black,         
   citecolor=black,         
   filecolor=black,         
   urlcolor=black,          
}
\title{Computing Individual Discrete Logarithms Faster in $\GF(p^n)$
  with the NFS-DL Algorithm
  \thanks{Copyright IACR 2015. This article is a minor revision of
    the ASIACRYPT 2015 final version. The version published by
    Springer-Verlag is available at
    \url{http://dx.doi.org/10.1007/978-3-662-48797-6_7}.} 
  \thanks{This research was partially funded by Agence %
    Nationale de la Recherche grant ANR-12-BS02-0001.}
  \thanks{Publisher version September, 7th 2015, revised May, 26th 2016.}
}
  \titlerunning{Computing Individual Discrete Logarithms in NFS-DL
    algorithm}

  \author{Aurore Guillevic \inst{1}\fnmsep \inst{2}}
  \institute{
     Institut National de Recherche en Informatique et en
     Automatique (INRIA) \\ Grace Team,
     Inria Saclay, Palaiseau, France \and 
     École Polytechnique/LIX, Palaiseau, France \\
     \email{guillevic@lix.polytechnique.fr}
   }

\begin{document}
\maketitle

\begin{abstract}
The Number Field Sieve (NFS) algorithm is the best known method to
compute discrete logarithms (DL) in finite fields
$\mathbb{F}_{p^n}$, with $p$ medium to large and $n \geq 1$ small. This algorithm
comprises four steps: polynomial selection, relation collection,
linear algebra and finally, individual logarithm computation. The
first step outputs two polynomials defining two number fields, and a
map from the polynomial ring over the integers modulo each of these
polynomials to $\mathbb{F}_{p^n}$. 
After the relation collection and linear algebra
phases, the (virtual) logarithm of a subset of elements in each number
field is known. Given the target element in $\mathbb{F}_{p^n}$, the fourth
step computes a preimage in one number field. If one can write the
target preimage as a product of elements of known (virtual) logarithm, 
then one can deduce the discrete logarithm of the target. 

As recently shown by the Logjam attack, this final step can be
critical when it can be computed very quickly.
But we realized that computing an individual DL is much slower in medium-
and large-characteristic non-prime fields $\mathbb{F}_{p^n}$ with $n \geq 3$,
compared to prime fields and quadratic fields $\mathbb{F}_{p^2}$. We optimize
the first part of individual DL: the \emph{booting step}, by reducing
dramatically the size of the preimage norm. 
Its smoothness probability is higher, hence the running-time of the
booting step is much improved. 
Our method is very efficient for small extension fields with $2 \leq
n \leq 6$ and applies to any $n > 1$, in medium and large characteristic.
\end{abstract}
\keywords{Discrete logarithm, finite field, number field sieve,
  individual logarithm.}

\section{Introduction}

\subsection{Cryptographic Interest}

Given a cyclic group $(G, \cdot)$
and a generator $g$ of $G$, the discrete logarithm (DL) of $x \in G$ is the
element $1 \leq a \leq \# G$ such that $x = g^a$. In well-chosen
groups, the exponentiation $(g,a) \mapsto g^a$ is very fast but
computing $a$ from $(g,x)$ is conjectured to be very 
difficult: this is the Discrete Logarithm Problem (DLP), at the heart of many
asymmetric cryptosystems. 
The first group proposed for DLP was the multiplicative group of a prime finite field. Nowadays, the
group of points of elliptic curves defined over prime fields are
replacing the prime fields for DLP-based cryptosystems. In
pairing-based cryptography, the finite fields are still used, because they
are a piece in the pairing mechanism. 
It is important in cryptography to know precisely the difficulty
of DL computation in the considered groups, to estimate the security of the
cryptosystems using them. 
Finite fields have a particularity: there exists a subexponential-time
algorithm to compute DL in finite fields of medium to large
characteristic: the Number Field Sieve (NFS).
In small characteristic, this is even better: a quasi-polynomial-time algorithm was
proposed very recently~\cite{EC:BGJT14}.

In May 2015, an international team of academic researchers revealed a
surprisingly efficient attack against a Diffie-Hellman key exchange in
TLS, the \emph{Logjam} attack~\cite{CCS:ABDGGH15}.
After a seven-day-precomputation stage (for relation
collection and linear algebra of NFS-DL algorithm), it was made
possible to compute any given
individual DL in about one minute, for each of the two
targeted 512-bit prime finite fields. This was fast enough for a
man-in-the-middle attack. 
This experience shows how critical it can be to be able to
compute individual logarithms very fast. 

Another interesting application for fast individual DL is \emph{batch-DLP}, and
\emph{delayed-target DLP}: in these
contexts, an attacker aims to compute several DL in the same
finite field. Since the costly phases of relation 
collection and linear algebra are only done one time for any fixed finite field, only
the time for one individual DL is multiplied by the number of
targets. This context usually arises in pairing-based
cryptography and in particular in broadcast protocols and traitor
tracing schemes, where a large number of DLP-based public/private key pairs are
generated. 
 The time to compute one
individual DL is important in this context, even if
parallelization is available.

\subsection{The Number Field Sieve Algorithm for DL in Finite Fields}
We recall that the NFS algorithm is
made of four steps: polynomial selection, relation collection,
linear algebra and finally, individual logarithm computation. \emph{This
last step is mandatory to break any given instance of a discrete
logarithm problem.} 
The polynomial selection outputs two irreducible polynomials $f$ and
$g$ defining two number fields $K_f$ and $K_g$. 
One considers the rings $R_f = \ZZ[x]/(f(x))$ and $R_g = \ZZ[x]/(g(x))$.
There exist two maps $\rho_f, \rho_g$ to $\FF_{p^n}$, as shown in the
following diagram. Moreover, the monic polynomial defining the finite
field is $\psi = \gcd(f,g) \bmod p$, of degree $n$.
\begin{figure}[hbt]
\centering
\begin{tikzpicture}[auto]
  \node (QQx) {$\ZZ[x]$};
  \node[below of=QQx, node distance=3em] (Middle) {};
  \node[left of=Middle, node distance=5em] (Kf)
    {$R_f = \ZZ[x]/(f(x))$};
  \node[right of=Middle, node distance=5em] (Kg) 
    {$\ZZ[y]/(g(y)) = R_g$};
  \path[draw, -latex'] (QQx) -- node{} (Kf);
  \path[draw, -latex'] (QQx) -- node{} (Kg);
  \node[below of=Middle, node distance=3.2em] (GFpn)
  {$\FF_{p^n} = \FF_p[z]/(\psi(z))$};

  \path[draw, -latex'] (Kf) -- (GFpn);
  \node[below of=Kf, node distance = 1.7em] (rho-f) {$\rho_f: x \mapsto z~$};
  \path[draw, -latex'] (Kg) -- (GFpn);
  \node[below of=Kg, node distance = 1.7em] (rho-g) {$~\rho_g: y \mapsto z$};
\end{tikzpicture}
\caption{NFS diagram}
\label{fig:NFS-diagram}
\end{figure}
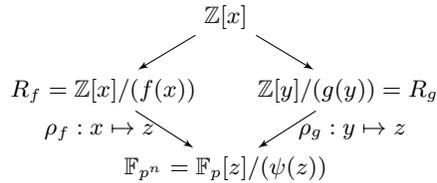
In the remaining of this paper, we will only use $\rho = \rho_f$,
$K = K_f$ and $R_f$. 
After the relation collection and linear algebra phases, the (virtual)
logarithm of a subset of 
elements in each ring $R_f, R_g$ is known. 
The individual DL step computes a preimage in one of the rings
$R_f, R_g$ of the target element in $\FF_{p^n}$. 
If one can write the target preimage as a product of
elements of known (virtual) logarithm, 
then one can deduce the individual DL of the target. 
The key point of individual DL computation is finding a smooth
 decomposition in small enough factors of the target preimage. 

\subsection{Previous Work on Individual Discrete Logarithm}
The asymptotic running time of NFS algorithm steps are estimated with
the $L$-function:
$$ L_Q[\alpha, c] = \exp \Bigl( \bigl(c+o(1)\bigr) (\log Q)^\alpha (\log \log
  Q)^{1-\alpha} \Bigr) \quad \mbox{ with } \alpha \in [0,1] \mbox{ and
  } c > 0~.$$
The $\alpha$ parameter measures the gap between polynomial time
($L_Q[\alpha=0, c]=\log^c Q$) and exponential time ($L_Q[\alpha=1,c] =
Q^c$). When $c$ is implicit,
or obvious from the context, we simply write $L_Q[\alpha]$. 
When the complexity relates to an algorithm for a prime field $\FF_p$,
we write $L_p[\alpha, c]$.
\paragraph{Large prime fields.}
Many improvements for computing discrete logarithms first concerned 
prime fields. 
The first subexponential DL algorithm in prime fields was due to Adleman
\cite{Adleman79} and had a complexity of $L_p[1/2,2]$.
In 1986, Coppersmith, Odlyzko and Schroeppel \cite{CopOdlSch86} introduced a
new algorithm (COS), of complexity $L_p[1/2, 1]$. They computed
individual DL \cite[\S 6]{CopOdlSch86} in $L_p[1/2, 1/2]$ in two
steps (finding a boot of medium-sized primes, then finding relations
of logarithms in the database for each medium prime).
In these two algorithms, the factor basis was quite large (the
smoothness bound was $L_p[1/2, 1/2]$ in both cases), providing a much
faster individual DL compared to relation collection and linear algebra.
This is where the common belief that individual
logarithms are easy to find (and have a negligible cost compared with the prior relation
 collection and linear algebra phases) comes from. 

In 1993, Gordon \cite{Gordon93} proposed the first version of NFS--DL
algorithm for prime fields $\FF_p$ with asymptotic complexity
$L_p[1/3, 9^{1/3}\simeq 2.08]$. 
However, with the $L_p[1/3]$ algorithm there are new difficulties,
among them the individual DL phase. 
In this $L_p[1/3]$ algorithm, many fewer logarithms of small elements are known,
because of a smaller smoothness bound (in $L_p[1/3]$ instead of
$L_p[1/2]$). The relation collection is shortened, explaining the
$L_p[1/3]$ running time.
But in the individual DL phase, since some non-small elements in the decomposition of the target have
an unknown logarithm, a dedicated
sieving and linear algebra phase is done for each of them. Gordon
estimated the running-time of individual DL computation to be
$L_p[1/3, 9^{1/3} \simeq 2.08]$, i.e. \emph{the same as the first two phases}.
In 1998, Weber \cite[\S 6]{EC:Weber98} compared the NFS--DL algorithm to
the COS algorithm 
for a 85 decimal digit prime and made
the same observation about individual DL cost.

In 2003, ten years after Gordon's algorithm, Joux and Lercier \cite{JoLe03} were the first to dissociate
in NFS relation collection plus linear algebra on one side and individual
DL on the other side. 
They used the \emph{special}-$q$ technique to find the logarithm of
medium-sized elements in the target decomposition.
In 2006, Commeine and Semaev \cite{PKC:ComSem06}
analyzed the Joux--Lercier method. They obtained an
asymptotic complexity of $L_p[1/3, 3^{1/3} \simeq 1.44]$ for computing individual
logarithms, independent of the relation collection and linear algebra
phases.
In 2013, Barbulescu \cite[\S4, \S 7.3]{Bar13thesis} gave a tight
analysis of the individual DL computation for prime fields,
decomposed in three steps: booting (also called smoothing), descent, and final
combination of logarithms. The booting step has an asymptotic
complexity of $L_p[1/3, 1.23]$ and the descent step of $L_p[1/3,
1.21]$. The final computation has a negligible cost.
\paragraph{Non-prime fields of medium to large characteristic.}
In 2006, Joux, Lercier, Smart and Vercauteren \cite{C:JLSV06} computed a
discrete logarithm in a cubic extension of a  prime field. They used the
special-$q$ descent technique again. 
They proposed for large characteristic fields an equivalent of the
\emph{rational reconstruction} technique for
prime fields and the \emph{Waterloo} algorithm \cite{C:BlaMulVan84} for small characteristic
fields, to improve the initializing step preceding the descent.
For DLs in prime fields, the target is an
integer modulo $p$. The rational reconstruction method outputs two
integers of half size compared to $p$, such that their quotient is
equal to the target element modulo $p$. 
Finding a smooth decomposition of the
target modulo $p$ becomes equivalent to 
finding a (simultaneous) smooth decomposition of two elements, each of
half the size. 
We explain their method (that we call the JLSV fraction method in the following) 
for extension fields in Sec.~\ref{sec:JLSV06 fraction}.
\paragraph{Link with polynomial selection.}
The running-time for finding a smooth decomposition depends on the
norm of the target preimage. The norm preimage depends on the
polynomial defining the number field. In particular, the smaller
the coefficients and degree of the polynomial, the smaller the preimage norm. Some
polynomial selection methods output polynomials that produce much
smaller norm.
That may be one of the reasons why the record computation of Joux
\textit{et al.}  \cite{C:JLSV06} used another polynomial 
selection method, whose first polynomial has 
 very small coefficients, and the second one has coefficients of size
$O(p)$. Thanks to the very small coefficients of the first polynomial,
their fraction technique
was very useful. Their polynomial selection technique is now superseded by their
JLSV\sub{1} method \cite[\S 2.3]{C:JLSV06} for larger values of $p$. 
As noted in \cite[\S 3.2]{C:JLSV06}, the fraction technique 
is useful in practice for small $n$. 
But for the JLSV\sub{1} method and $n \geq 3$, this is already too
slow (compared to not using it). 
In 2008, Zajac~\cite{Zajac2008PhD} implemented the NFS-DL algorithm
for computing DLs in $\FF_{p^6}$ with $p$ of 40 bits
(12 decimal digits (dd), i.e. $\FF_{p^6}$ of 240 bits or 74 dd). 
He used the methods described in \cite{C:JLSV06},
with a first polynomial with very small coefficients and a second one with
coefficients in $O(p)$. In this case, individual DL computation
was possible (see the well-documented~\cite[\S 8.4.5]{Zajac2008PhD}).
In 2013, Hayasaka, Aoki, Kobayashi and Takagi~\cite{HAKT13} computed a
DL in $\FF_{p^{12}}$ with $p = 122663$ ($p^n$ of 203
bits or 62 dd). 
We noted that all these records used the same polynomial
selection method, so that one of the polynomials has very small coefficients
(e.g. $f = x^3 +x^2 - 2 x -1$) whereas the second one has coefficients
in $O(p)$. 

In 2009, Joux, Lercier, Naccache and Thomé~\cite{IMA:JLNT09} proposed
an attack of DLP in a protocol context. The
relation collection is sped up with queries to an oracle. They
wrote in \cite[\S B]{IMA:JLNT09} an extended analysis of individual
DL computation. In their case, the individual logarithm phase
of the  NFS-DL algorithm has a running-time of $L_Q[1/3, c]$
where $c = 1.44$ in the large characteristic case, and $c = 1.62$ in
the medium characteristic case.
In 2014, Barbulescu and Pierrot \cite{BarPie2014} presented a multiple number field
sieve variant (MNFS) for extension fields, based on Coppersmith's
ideas~\cite{JC:Coppersmith93}. The individual 
logarithm is studied in \cite[\S A]{BarPie2014}. They also used a \emph{descent}
technique, for a global estimated running time in $L_Q[1/3,
(9/2)^{1/3}]$, with a constant $c \approx 1.65$. 
Recently in 2014, Barbulescu, Gaudry, Guillevic and Morain
\cite{BGGM14-DL-record,EC:BGGM15} announced 160 and 180 decimal digit discrete
logarithm records in quadratic fields. They also used a technique
derived from the JLSV fraction method and a special-$q$
descent technique, but did not give an asymptotic running-time. 
It appears that this technique becomes inefficient as soon as $n$ = 3 or
4.
\paragraph{Overview of NFS-DL asymptotic complexities.}
The running-time of the relation collection step and the individual
DL step rely on the smoothness probability of integers. 
An integer is said to be 
$B$-smooth if all its prime divisors are less than $B$. An ideal in a
number field is said to be $B$-smooth if it factors into prime ideals whose
norms are bounded by $B$. 
Usually, the relation collection and the linear algebra are balanced,
so that they have both the same dominating asymptotic complexity. 
The NFS algorithm for DL in prime and large characteristic fields has
a dominating complexity of $L_Q[1/3, (\frac{64}{9})^{1/3}\simeq 1.923]$. 
For the individual DL in a prime field $\FF_p$, the norm of the target preimage in
the number field is bounded by $p$. This bound gives the
running time of this fourth step (much smaller than relation
collection and linear algebra).
Finding a smooth decomposition of the preimage and
computing the individual logarithm (see \cite{PKC:ComSem06}) has
complexity $L_p[1/3, c]$ with $c= 1.44$, and $c=1.23$ with the improvements of
\cite{Bar13thesis}. The booting step is dominating.
In large characteristic fields,
the individual DL has a complexity of $L_Q[1/3,1.44]$, dominated by the
booting step again (\cite[\S B]{IMA:JLNT09} for JLSV\sub{2},
Table~\ref{tab: new booting step cpx} for gJL). 

In generic medium characteristic fields, the complexity of the NFS algorithm
is $L_Q[1/3, (\frac{128}{9})^{1/3} = 2.42]$ with the JLSV\sub{1} method proposed in
\cite[\S 2.3]{C:JLSV06}, $L_Q[1/3, (\frac{32}{3})^{1/3} = 2.20]$ with
the Conjugation method \cite{EC:BGGM15}, and
$L_Q[1/3, 2.156]$ with the MNFS version \cite{EC:Pierrot15}.
We focus on the individual DL step with the JLSV\sub{1} and Conjugation methods.
In these cases, the preimage norm bound is in fact
much higher than in prime fields. Without any improvements, 
the dominating booting step
has a complexity of $L_Q[1/3, c]$ with $c = 1.62$ \cite[\S
C]{IMA:JLNT09} or $c=1.65$ \cite[\S A]{BarPie2014}.
However, this requires to sieve over ideals of degree $1 < t < n$.
For the Conjugation method, this is worse: the booting step has a
running-time of $L_Q[1/3, 6^{1/3}\simeq 1.82]$ 
(see our computations in Table~\ref{tab: new booting step cpx}).
Applying the JLSV fraction method lowers the norm bound to $O(Q)$
for the Conjugation method. 
The individual logarithm in this case has complexity
$L_Q[1/3, 3^{1/3}]$ as for prime fields (without the improvements of
\cite[\S 4]{Bar13thesis}). 
However, this method is not suited for number fields generated with the
JLSV\sub{1} method, for $n \geq 3$. 

\subsection{Our Contributions}
\label{sec:contributions}
In practice, we realized that the JLSV fraction method  which seems interesting and
sufficient because of the $O(Q)$ bound, is in fact not convenient for the
gJL and Conjugation methods for $n$ greater than 3. The preimage norm
is much too large, so finding a smooth factorization is too slow by an order of
magnitude. 
We propose a way to lift the target from the finite field to the
number field, such that the norm is strictly smaller than $O(Q)$ for
the gJL and Conjugation methods:
\begin{theorem} \label{th: smaller norm bound}
Let $n > 1$ and $s \in \FF_{p^n}^{*}$ a random element (not in a proper subfield
of $\FF_{p^n}$). We want to compute
its discrete logarithm modulo $\ell$, where $\ell \mid \Phi_n(p)$,
with $\Phi_n$ the $n$-th cyclotomic polynomial.  
Let $K_f$ be the number field given by a polynomial selection method, 
whose defining polynomial has the smallest coefficient size, and $R_f
= \ZZ[x]/(f(x))$. 
Then there exists a preimage $\lift{r}$ in $R_f$ of some $r \in \FF_{p^n}^*$, such that
$\log \rho(\lift{r}) \equiv \log s \pmod \ell$ and 
such that the norm of $\lift{r}$ in $K_f$ is bounded by $O(Q^e)$, where $e$ is equal to
\begin{enumerate}
  \item $1-\frac{1}{n}$ for the gJL and Conjugation methods;
  \item $\frac{3}{2} - \frac{3}{2n}$ for the JLSV\sub{1} method;
  \item  $1 - \frac{2}{n}$ for the Conjugation method, if $K_f$ has a well-chosen
    quadratic subfield satisfying the conditions of Lemma~\ref{lemma:Fp2 simplification};
  \item $\frac{3}{2} -\frac{5}{2n}$ for the JLSV\sub{1} method, if $K_f$ has a well-chosen
    quadratic subfield satisfying the conditions of Lemma~\ref{lemma:Fp2 simplification}.
\end{enumerate}
\end{theorem}
Our method reaches the optimal bound of $Q^{\varphi(n)/n}$, with
$\varphi(n)$ the Euler totient function, 
 for $n=2,3,4,5$ combined with the gJL or
the Conjugation method.
We show that our method provides a dramatic improvement for individual logarithm
computation for small $n$: 
the running-time of the booting step (finding boots) is
$L_Q[1/3, c]$ with $c = 1.14$ for $n=2,4$, $c = 1.26$ for $n=3,6$ and
$c = 1.34$ for $n=5$. 
It generalizes to any $n$, so that the norm is always smaller than
$O(Q)$ (the prime field case), hence the booting step running-time in $L_Q[1/3, c]$
always satisfies $c < 1.44$  for the two
state-of-the-art variants of NFS for extension fields (we have $c=1.44$ for prime fields).
For the JLSV\sub{1} method, this bound is satisfied for $n=4$, where
we have $c=1.38$ (see Table~\ref{tab: new booting step cpx}).
\subsection{Outline}
We select three polynomial selection methods involved for NFS-DL in
generic extension fields and recall their properties in
Sec.~\ref{subsec:polyselect}.  We recall
 a commonly used bound on the norm of an element in a number
field (Sec.~\ref{subsec: norm bound}). We present in 
Sec.~\ref{sec:JLSV06 fraction} a generalization
of the JLSV fraction method of \cite{C:JLSV06}.
In Sec.~\ref{sec: asympt cpx} we give a proof 
of the booting step complexity stated in Lemma~\ref{th: asympt cpx}. 
We sketch in Sec.~\ref{subsec:special-q descent cpx} the
special-$q$ descent technique and list the asymptotic complexities
found in the literature according to the polynomial selection methods.
We present in Sec.~\ref{sec: preimage} 
 our main idea to reduce the norm of the preimage in the number field, by reducing
 the preimage coefficient size with the LLL algorithm. 
We improve our technique in Sec.~\ref{sec: preimage smaller norm}
by using a quadratic subfield when available, to finally complete the
proof of Theorem~\ref{th: smaller norm bound}.
We provide practical examples in Sec.~\ref{sec: practical examples}, for 180 dd finite fields in 
Sec.~\ref{subsec: 180dd examples} and we give our running-time
experiments for a 120 dd
finite field $\FF_{p^4}$ in Sec.~\ref{subsec:Fp4 120dd}.
\section{Preliminaries}
\label{sec:preliminaries}
%

We recall an important property of the LLL algorithm~\cite{LLL82}
that we will widely use in this paper.
Given a lattice $\mathcal{L}$ of $\ZZ^n$ defined by a basis given in an $n \times n$ matrix $L$,  
and parameters $\frac{1}{4} < \delta <1$, $\frac{1}{2} < \eta < \sqrt{\delta}$, the LLL algorithm outputs a $(\eta, \delta)$-\emph{reduced basis} of the lattice.
the coefficients of the first (shortest) vector  
are bounded by
$$(\delta - \eta^2)^{\frac{n-1}{4}} \det(L)^{1/n} ~.$$
With $(\eta, \delta)$ close to $(0.5, 0.999)$ (as in NTL or magma),
the approximation factor $C = (\delta - \eta^2)^{\frac{n-1}{4}}$ is
bounded by $1.075^{n-1}$ (see \cite[\S 2.4.2]{Chen13})). 
Gama and Nguyen experiments \cite{EC:GamNgu08} on numerous random
lattices showed that on average, $C \approx 1.021^n$.
In the remaining of this paper, we will simply denote by $C$ this LLL
approximation factor.

\subsection{Polynomial Selection Methods}
\label{subsec:polyselect}

We will study the booting step of the NFS algorithm with these three
polynomial selection methods:
\begin{enumerate}
\item the Joux--Lercier--Smart--Vercauteren (JLSV\sub{1}) method
  \cite[\S 2.3]{C:JLSV06};
\item the generalized Joux--Lercier (gJL) method \cite[\S 2]{Mat06}, 
\cite[\S 3.2]{EC:BGGM15};
\item the Conjugation method \cite[\S3.3]{EC:BGGM15}.
\end{enumerate}
In a non-multiple NFS version, the JLSV\sub{2} \cite[\S 2.3]{C:JLSV06} and gJL methods have the best asymptotic running-time in the large
characteristic case, while the Conjugation method holds the best one in the
medium characteristic case. However for a record computation in $\FF_{p^2}$, the Conjugation
method was used \cite{EC:BGGM15}. For medium characteristic fields of
record size (between 150 and 200 dd), is seems also that the JLSV\sub{1}
method could be chosen (\cite[\S 4.5]{EC:BGGM15}). 
Since the use of each method is not fixed in practice, 
we study and compare the three above methods for the individual
logarithm step of NFS. We recall now the construction and properties
of these three methods.

\paragraph{Joux--Lercier--Smart--Vercauteren (JLSV\sub{1}) Method.}
This method was introduced in 2006. We describe it in
Algorithm~\ref{alg:JLSV1}. 
The two polynomials $f,g$ have degree $n$ and coefficient size
$O(p^{1/2})$. We set $\psi = \gcd(f,g) \bmod p$ monic of degree
$n$. We will use $\psi$ to represent
the finite field extension $\FF_{p^n} = \FF_p[x]/(\psi(x))$. 

\begin{algorithm}[htbp]
\DontPrintSemicolon
\caption{Polynomial selection with the JLSV\sub{1} method~\cite[\S 2.3]{C:JLSV06}}
\label{alg:JLSV1}
\KwIn{ $p$ prime and $n$ integer}
\KwOut{ $f,g,\psi$ with $f,g\in\Z[x]$ irreducible and $\psi=\gcd(f
\bmod p,g \bmod p)$ in $\FF_p[x]$ irreducible of degree~$n$
}
Select $f_1(x), f_0(x)$, two polynomials with small integer coefficients,
$\deg f_1 < \deg f_0 = n$ \;
\Repeat {$f=f_0+ y f_1$ is irreducible in $\FF_p[x]$}
{choose $y \approx \lceil \sqrt{p}\rceil$\;}
 $(u,v)\gets$ a rational reconstruction of $y$ modulo $p$  \;
 $g\gets vf_0 +u f_1 $ \;
 \Return {$(f, g, \psi=f\bmod p)$}
\end{algorithm}

\paragraph{Generalized Joux--Lercier (gJL) Method.}
This method was independently proposed in \cite[\S 2]{Mat06} and
\cite[\S 8.3]{Bar13thesis} (see also \cite[\S 3.2]{EC:BGGM15}). This is a
generalization of the Joux--Lercier method \cite{JoLe03} for prime fields.
We sketch this method in Algorithm~\ref{alg:gJL}. 
The coefficients of $g$ have size $O(Q^{1/(d+1)})$ and those of $f$
have size $O(\log p)$, with $\deg g = d \geq n$ and $\deg f = d+1$.

\begin{algorithm}[htbp]
\DontPrintSemicolon
\caption{Polynomial selection with the gJL method
}
\label{alg:gJL}
\KwIn{ $p$ prime, $n$ integer and $d \geq n$ integer}
\KwOut{ $f,g,\psi$ with $f,g\in\Z[x]$ irreducible and $\psi=\gcd(f
\bmod p,g \bmod p)$ in $\FF_p[x]$ irreducible of degree~$n$
}
Choose a polynomial $f(x)$ of degree $d+1$ with small integer
coefficients which has a monic irreducible factor $\psi(x) = \psi_0 +
\psi_1 x+ \cdots + x^n$ of degree $n$ modulo $p$ \;
Reduce the following matrix using LLL
$$
M = \begin{bmatrix}
p          &           &           &           &       & \\
           & \ddots    &           &           &       & \\
           &           & p         &           &       & \\
\psi_{0}&\psi_{1}& \cdots    &1          &       & \\
           & \ddots    & \ddots    &           &\ddots & \\
           &           &\psi_{0}&\psi_{1}&\cdots &1\\
\end{bmatrix}
\begin{array}{l}
   \left \rbrace \begin{array}{l}
       \phantom{\ddots} \\
       \deg \psi = n  \\
       \phantom{I} \\
   \end{array}\right. \\
   \left \rbrace \begin{array}{l}
        \phantom{\ddots} \\
         d+1-n  \\
        \phantom{I} \\
   \end{array}\right.
\end{array}~, 
\mbox{ to get }
\LLL(M) = \begin{bmatrix}
        g_0 & g_1 & \cdots & g_{d} \\
	&    &        &        \\
        &    &        &        \\
        &     \multicolumn{2}{c}{*}    & \\
	&    &        &         \\
        &    &        &         \\
\end{bmatrix}~.
$$
 \Return {$(f, g = g_0 + g_1 x + \cdots + g_{d} x^{d}, \psi)$}
\end{algorithm}

\paragraph{Conjugation Method.}
This method was published in \cite{EC:BGGM15} and used for the discrete
logarithm record in $\FF_{p^2}$, with $f= x^4+1$. 
The coefficient size of $f$ is in $O(\log p)$ and the coefficient size
of $g$ is in $O(p^{1/2})$. 
We describe it in Algorithm~\ref{alg:conjugation}.

\begin{algorithm}[hbtp]
\DontPrintSemicolon
\caption{Polynomial selection with the Conjugation method \cite[\S 3.3]{EC:BGGM15}}
\label{alg:conjugation}
\KwIn{ $p$ prime and $n$ integer}
\KwOut{ $f,g,\psi$ with $f,g\in\Z[x]$ irreducible and $\psi=\gcd(f
\bmod p,g \bmod p)$ in $\FF_p[x]$ irreducible of degree~$n$
}
\Repeat {$P_y(Y)$ has a root $y$ in $\FF_p$ and
$\psi(x)=g_0(x)+y g_1(x)$ is irreducible in $\FF_p[x]$}
{
Select $g_1(x), g_0(x)$, two polynomials with small integer coefficients,
$\deg g_1 < \deg g_0 = n$ \;
Select $P_y(Y)$ a quadratic, monic, irreducible polynomial over
$\Z$ with small coefficients \;}
 $f\gets \Reslt_Y(P_y(Y), g_0(x) + Yg_1(x))$ \;
 $(u,v)\gets$ a rational reconstruction of $y$ \;
 $g\gets vg_0 +u g_1 $ \;
 \Return {$(f, g, \psi)$}
\end{algorithm}

\begin{table}[htbp]
\centering
\caption{Properties: degree and coefficient size of the three
  polynomial selection methods for NFS-DL in $\FF_{p^n}$. The
  coefficient sizes are in $O(X)$. To lighten the notations, we
  simply write the $X$ term.} 
\label{tab: polyselect comparison}
$\begin{array}{|c ||c | c |c| c|}
\hline\mbox{method}& \deg f & \deg g & \|f\|_\infty  & \|g\|_\infty   \\
\hline
\hline\mbox{JLSV}_1&  n     &   n    & Q^{1/2n}    & Q^{1/2n}     \\
\hline\mbox{gJL}   &  ~d+1 > n~ & ~ d \geq n ~& ~\log p~ & ~Q^{1/(d+1)}~ \\ 
\hline\mbox{Conjugation} & 2n & n    & \log p     & Q^{1/2n}     \\
\hline
 \end{array}$
\end{table}

\subsection{Norm Upper Bound in a Number Field}
\label{subsec: norm bound}
In Sec.~\ref{sec: preimage} we will compute the norm of an element $s$
in a number field $K_f$. We will need an upper bound of this norm. 
For all the polynomial selection methods chosen, $f$ is monic, whereas
$g$ is not. We remove the leading coefficient of $f$ from any formula involved
with a monic $f$. 
So let $f$ be a monic irreducible polynomial over $\QQ$ and let $K_f =
\QQ[x]/(f(x))$ a number field. Write $s \in K_f$ as a polynomial in
$x$, i.e. $s 
= \sum_{i=0}^{\deg f-1} s_i x^i$. 
The norm is defined by a resultant computation:
$$\Norm_{K_f/\QQ}(s) = \Reslt(f, s) ~.$$
We use Kalkbrener's bound \cite[Corollary 2]{Kal97} for an upper bound:
$$|\Reslt(f,s)|\leq \kappa(\deg f,\deg s) \cdot \norm{f}_\infty^{\deg s}
\norm{s}_\infty^{\deg f},$$
where $\kappa(n,m)=\binom{n+m}{n}\binom{n+m-1}{n}$, and
$\norm{f}_\infty = \max_{0 \leq i \leq \deg f} |f_i|$ the absolute
value of the greatest coefficient.
An upper bound for $\kappa(n,m)$ is $(n+m)!$.
We will use the following bound in Sec.~\ref{sec: preimage}:
\begin{equation}\label{eq:norm bound}
|\Norm_{K_f/\QQ}(s)| \leq (\deg f + \deg s)! \norm{f}_\infty^{\deg s}
\norm{s}_\infty^{\deg f} ~.
\end{equation}

\subsection{Joux--Lercier--Smart--Vercauteren Fraction Method}
\label{sec:JLSV06 fraction}
\begin{notation}{Row and column indices.}
In the following, we will define matrices of size $d \times d$, with $d
\geq n$. For ease of notation, we will index the rows and columns from
$0$ to $d-1$ instead of 1 to $d$, so that the $(i+1)$-th row at
index $i$, $L_i = [L_{ij}]_{0 \leq j \leq d-1}$, can be written
in polynomial form $\sum_{j=0}^{d-1} L_{ij} x^j$, and the column
index $j$ coincides with the degree $j$ of $x^j$.
\end{notation}
In 2006 was proposed in \cite{C:JLSV06} a method to generalize to
non-prime fields the rational reconstruction method used for prime
fields.
In the prime field setting, the target is an
integer modulo $p$. The rational reconstruction method outputs two
integers of half size compared to $p$ and such that their quotient is
equal to the target element modulo $p$. Finding a smooth decomposition of the
target modulo $p$ becomes equivalent to finding at the same time a smooth decomposition
of two integers of half size each. 

To generalize to extension fields, one writes the target preimage
as a quotient of two number field elements, each with a smaller norm
compared to the original preimage. 
 We denote by $s$ the target in the finite 
field $\FF_{p^n}$ and by $\lift{s}$ a preimage (or
lift) in $K$. 
Here is a first very simple preimage choice.
Let $\FF_{p^n} = \FF_p[x]/(\psi(x))$ and $s = \sum_{i=0}^{\deg s} s_i
x^i \in \FF_{p^n}$, with $\deg s < n$. 
We lift the coefficients $s_i \in \FF_p$
to $\lift{s}_i \in \ZZ$ then we set a preimage of $s$ in the number
field $K$ to be
$$\lift{s} = \sum_{i=0}^{\deg s} \lift{s}_i X^i ~,$$ 
with $X$ such that $K = \QQ[X]/(f(X))$. 
(We can also write $\lift{s} = \sum_{i=0}^{\deg s} \lift{s}_i \alpha^i$, 
with $\alpha$ a root of $f$ in the number field: $K = \QQ[\alpha]$).
We have $\rho(\lift{s}) = s$.

Now LLL is used to obtain a quotient whose numerator and denominator
have smaller coefficients. We present here the lattice used
with the JLSV\sub{1} polynomial selection method. The number field
$K$ is of degree $n$. We define a lattice of dimension $2n$. For the
corresponding matrix, each column of the left half corresponds to a
power of $X$ in the numerator; each column of the right half corresponds to a
power of $X$ in the denominator.
The matrix is
$$
\begin{array}{cl}
 L = & \begin{bmatrix}
{p}   &       &             &      &             &   \\
          & {\ddots}&             &      &             &   \\
          &       & {p}           &      &             &   \\
{\lift{s}_0} &{\ldots} &{\lift{s}_{n-1}}       & {1 {\pM}}    &             &   \\
{\vdots}    &       & {\vdots}      &      &      & {\ddots}    &   \\
\mcm{3}{l}{{{\lift{s} x^{n-1} \pmod \psi }} }& &      &       & { {\pM} 1} \\
\end{bmatrix} 
\underset{2n\times 2n~~~~}{
  \begin{matrix}
    \sss{0} \\ \sss{\vdots} \\ \sss{n-1} \\ \sss{n} \\ \sss{\vdots} \\
 \sss{2n-1} \\ 
  \end{matrix}
} 
\end{array} $$
The first $n$ coefficients
of the output vector, $u_0, u_1, \ldots, u_{n-1}$ give a numerator
$u$ and the last $n$ coefficients give a denominator $v$, so that
$\lift{s} = a \frac{u(X)}{v(X)}$ with $a$ a scalar in $\QQ$. The
coefficients $u_i, v_i$ are bounded by 
$ \|u\|_\infty, ~ \|v\|_\infty \leq C p^{1/2} $
since the matrix determinant is $\det L = p^n$ and the matrix is of
size $2n \times 2n$. 
However the product of the norms of each $u, v$ in the number field
$K$ will be much larger than the norm of the single element $\lift{s}$
because of the large coefficients of $f$ in the norm formula.
We use formula \eqref{eq:norm bound} to estimate this bound:
$$ \Norm_{K/\QQ}(u) \leq 
\|u\|_{\infty}^{\deg f} \|f\|_\infty^{\deg u} =
O(p^{\frac{n}{2}} p^{\frac{n-1}{2}}
) = O(p^{n-\frac{1}{2}}) = O(Q^{1- \frac{1}{2n}})$$
and the same for $\Norm_{K/\QQ}(v)$, hence the product of the two
norms is bounded by $O(Q^{2 - \frac{1}{n}})$. The norm of $\lift{s}$ is
bounded by $\Norm_{K/\QQ}(\lift{s}) \leq p^n p^{\frac{n-1}{2}} =
Q^{\frac{3}{2} - \frac{1}{2n}}$ which is much smaller whenever $n \geq
3$. Finding a
smooth decomposition of $u$ and $v$ at the same time will be much
slower than finding one for $\lift{s}$ directly, for large $p$ and $n
\geq 3$. This is mainly because of the large coefficients of $f$ (in
$O(p^{1/2})$). 

\subsubsection{Application to gJL and Conjugation Method.}
The method of \cite{C:JLSV06} to improve the smoothness of the target
norm in the number field $K_f$ has an advantage for the gJL and
Conjugation methods. 
First we note that the number field degree is larger than $n$: this is
$d+1 \geq n+1$ for the gJL method and $2n$ for the Conjugation
method. For ease of notation, we denote by $d_f$ the degree of
$f$. We define a lattice of dimension $2 d_f$. Hence there is more
place to reduce the coefficient size of the target $\lift{s}$.

We put $p$ on the diagonal of the first $n-1$ rows, then $x^i\psi(x)$
coefficients from row $n$ to $d_f -1$, where $0 \leqslant i < d_f-1$
($\psi$ is of degree $n$ and has $n+1$ coefficients). 
The rows from index $d_f$ to $2 d_f$ are filled with $X^i \lift{s} \bmod
f$ (these elements have $d_f$ coefficients).
We obtain a triangular matrix $L$.
$$ L =  \begin{bmatrix}
p          &      &             &         &             &      &      &      &        &       \\
           &\ddots&             &         &             &      &      &      &        &       \\
           &      & p           &         &             &      &      &      &        &       \\
\psi_{0}&\cdots&\psi_{n-1}& 1       &             &      &      &      &        &       \\
           &\ddots&             &\ddots   &\ddots       &      &      &      &        &       \\
           &      &\psi_{0}  &\cdots   &\psi_{n-1}&1     &      &      &        &       \\
\lift{s}_0  &\ldots&\lift{s}_{n-1}&         &             &      & 1\pM &      &        &       \\
\vdots     &      &             &         &             &      &      &      & \ddots &       \\
\mcm{3}{l}{ X^{d_f-1} \lift{s} \bmod f} &  &             &      &      &      &        & \pM 1 \\
\end{bmatrix}
\underset{2d_f \times 2d_f}{
  \begin{matrix}
    \sss{0} \\ \sss{\vdots} \\ \sss{n-1} \\ \sss{n} \\ \sss{\vdots} \\
    \sss{d_f-1} \\ \sss{d_f} \\ \sss{\vdots} \\
    \sss{2d_f} \\ 
  \end{matrix}
} $$
Since the determinant is $ \det L = p^n$ and the matrix of dimension
$2d_f\times 2 d_f$, the coefficients obtained with LLL will be bounded
by $ C p^{\frac{n}{2 d_f}}$. The norm of the numerator or the 
denominator (with $\lift{s} = u(X)/v(X) \in K_f$) is bounded by
$$ \Norm_{K_f/\QQ}(u) \leq \|u\|_\infty ^{\deg f} \|f\|_\infty^{\deg
  u} = O(p^{n/2}) = O(Q^{1/2})~.$$
The product of the two norms will be bounded by $O(Q)$ hence we will
have the same asymptotic running time as for prime fields, for finding a smooth
decomposition of the target in a number field obtained with the gJL or
Conjugation method. We will show in Sec.~\ref{sec: preimage} that we
can do even better.

\section{Asymptotic Complexity of  Individual DL
  Computation}

\subsection{Asymptotic Complexity of Initialization or Booting Step}
\label{sec: asympt cpx}

In this section, we prove the following lemma on the  booting step running-time to find a
smooth decomposition of the norm preimage. This was already proven
especially for an initial norm bound of $O(Q)$. We state it in the
general case of a norm bound of $Q^e$. The smoothness
bound $B=L_Q[2/3,\gamma]$ used here is not the same as for the relation collection step, where
the smoothness bound was $B_0 = L_Q[1/3, \beta_0]$. Consequently, the
 special-$q$ output in the booting step will be bounded by $B$.
\begin{lemma}[Running-time of ${B}$-smooth decomposition]
\label{th: asympt cpx}
Let $s \in \FF_Q$ of order $\ell$. Take at random $t \in [1, \ell-1]$ and assume that
the norm $S_t$ of a preimage of $s^t \in \FF_Q$, in the number field $K_f$, is bounded
by $Q^e = L_Q[1, e]$. Write ${B} = L_Q[\alpha_{B}, \gamma]$ the smoothness bound
for $S_t$. Then the lower bound of the expected running time for finding $t$
s.t. the norm $S_t$ of $s^t$ is ${B}$-smooth is $L_Q[1/3, (3e)^{1/3}]$, obtained with
$\alpha_{B} = 2/3$ and  $\gamma = (e^2/3)^{1/3}$. 
\end{lemma}
First, we need a result on smoothness probability.
We recall the definition of $B$-smoothness already stated in
Sec.~\ref{sec:contributions}: 
an integer $S$ is $B$-smooth if and only if all its prime divisors are
less than or equal to $B$.
We also recall the
$L$-notation widely used for sub-exponential asymptotic
complexities:
$$ L_Q[\alpha, c] = \exp \Bigl( \bigl(c+o(1)\bigr) (\log Q)^\alpha (\log \log
  Q)^{1-\alpha} \Bigr) \quad \mbox{ with } \alpha \in [0,1] \mbox{ and
  } c > 0~.$$
The Canfield--Erd\H{o}s--Pomerance \cite{CEP83} theorem provides a useful
result to measure smoothness probability:
\begin{theorem}[$B$-smoothness probability]
\label{cor:B-smooth pr}
Suppose $0 < \alpha_B < \alpha_S \leq 1$, $\sigma >0$, and $\beta > 0$
are fixed. 
For a random integer $S$ bounded by $L_Q[\alpha_S, \sigma]$ and a smoothness bound $B =
L_Q[\alpha_B, \beta]$, the probability that $S$ is $B$-smooth is 
\begin{equation}
\Prb(S\mbox{ is } B\mbox{-smooth}) = L_Q\Bigl[\alpha_S -
\alpha_B, -(\alpha_S -\alpha_B) \frac{\sigma}{\beta} \Bigr] 
\end{equation}
for $Q \to \infty$.
\end{theorem}

We prove now the Lemma~\ref{th: asympt cpx} that states the
running-time of individual logarithm when the norm of the target in a
number field is bounded by $O(Q^{e})$.

\begin{proof}[of Lemma~\ref{th: asympt cpx}]
From Theorem~\ref{cor:B-smooth pr}, the probability that $S$ bounded by
$Q^e = L_Q[1, e]$ is ${B}$-smooth with ${B} = L_Q[\alpha_{B}, \gamma]$ is 
$\Prb(S\mbox{ is }{B}\mbox{-smooth}) = L_Q \bigl[1-\alpha_{B}, -(1-\alpha_{B})\frac{e}{\gamma}\bigr]$. We
assume that a ${B}$-smoothness test with ECM takes time 
$L_{B}[1/2, 2^{1/2}] = L_Q [\frac{\alpha_{B}}{2}, (2 \gamma
\alpha_{B})^{1/2}]$. 
The running-time for finding a ${B}$-smooth decomposition of $S$ is the
ratio of the time per test (ECM cost) to the ${B}$-smoothness
probability of $S$: 
$$ L_Q\Bigl[ \frac{\alpha_{B}}{2}, (2 \gamma \alpha_{B})^{1/2} \Bigr] 
L_Q \Bigl[ 1-\alpha_{B}, (1-\alpha_{B})\frac{e}{\gamma} \Bigr] ~.$$
We optimize first the $\alpha$ value, so that $\alpha \leq 1/3$ 
(that is, not exceeding the $\alpha$ of the two previous steps of the
NFS algorithm): 
$ \max(\alpha_{B}/2, 1 - \alpha_{B}) \leq \frac{1}{3} ~. $
This gives the system
$\left \lbrace 
\begin{array}{l}
\alpha_{B} \leq 2/3 \\
\alpha_{B} \geq 2/3 \\
\end{array}
\right.  $
So we conclude that 
$\alpha_{B} = \frac{2}{3}$.
The running-time for finding a ${B}$-smooth decomposition of $S$ is
therefore 
$$L_Q\Bigl[1/3, \Bigl(\frac{4}{3}\gamma\Bigr)^{1/2} + \frac{e}{3\gamma}\Bigr]~. $$
The minimum\footnote{One computes the derivative of the function
  $h_{a,b}(x) = a \sqrt{x} + \frac{b}{x}$: this is $h^{'}_{a,b}(x) =
  \frac{a}{2\sqrt{x}} - \frac{b}{x^2}$ and find that the minimum of $h$
  for $x>0$ is $h_{a,b}((\frac{2b}{a})^{2/3}) = 3(\frac{a^2b}{4})^{1/3}$. With $a =
  2/3^{1/2}$ and $b = e/3$, we obtain the minimum: $h((\frac{e^2}{3})^{1/3}) =
  (3e)^{1/3}$.
} 
of the function $\gamma \mapsto (\frac{4}{3}\gamma)^{1/2} + \frac{e}{3\gamma}$ is
$(3e)^{1/3}$, corresponding to $\gamma = (e^2/3)^{1/3}$, which yields
our optimal running time, together with the special-$q$ bound $B$: 
$$L_Q\Bigl[1/3, (3e)^{1/3}\Bigr]\quad \mbox{ with } q \leq B = L_Q\Bigl[2/3, (e^2/3)^{1/3}\Bigr]~.$$
\qed
\end{proof}

\subsection{Running-Time of Special-$q$ Descent}
\label{subsec:special-q descent cpx}

The second step of the individual logarithm computation is the
\emph{special-$q$ descent}. This consists in computing the
logarithms of the medium-sized elements in the factorization of the
target in the number field. The first special-$q$ is of order
$L_Q[2/3, \gamma]$ (this is the boot obtained in the initialization step) and is
the norm of a 
degree one prime ideal in the number field where the booting step
was done (usually $K_f$). 
The idea is to \emph{sieve}  
over linear combinations of degree one
ideals, in $K_f$ and $K_g$ at the same time, whose norms for one side
will be multiples of $q$ by construction, in order 
to obtain a relation involving a degree one prime ideal of norm $q$
and other degree one prime ideals of norm strictly smaller than $q$. 

Here is the common way to obtain such a relation. Let $\qq$ be  a
degree one prime ideal of $K_f$, whose norm is $q$. We can write $\qq =
\langle q, r_q\rangle$, with $r_q$ a root of $f$ modulo $q$ (hence $|r_q| < q$). 
We need to compute two
ideals $\qq_1, \qq_2 \in K_f$ whose respective norm is a multiple of $q$, and
sieve over $a \qq_1 + b \qq_2$. The
classical way to construct these two ideals is to reduce the
two-dimensional lattice generated by $q$ and $r_q - \alpha_f$,
i.e. to compute $\mbox{LLL}\left(\begin{bmatrix} q & 0 \\ -r & 1 \\ \end{bmatrix}\right) = $
$\begin{bmatrix} u_1 & v_1 \\ u_2 & v_2 \\ \end{bmatrix} $
to obtain two degree-one ideals $ u_1 + v_1
\alpha_f, u_2 + v_2 \alpha_f$ with shorter coefficients. One sieves over
 $\rrr_f = (a u_1 + b u_2) + (a v_1 + b v_2) \alpha_f $ 
and  $\rrr_g = (a u_1 + b u_2) + (a v_1 + b v_2) \alpha_g $.
The new ideals obtained in the relations will be treated as new special-$q$s
until a relation of ideals of norm bounded by $B_0$ is found, where
$B_0$ is the bound on the factor basis, so that the individual logarithms
are finally known. 
 The sieving is done in three stages, for the
three ranges of parameters. 
\begin{enumerate}
  \item For $q = L_Q[2/3, \beta_1]$: large special-$q$;
  \item For $q = L_Q[\lambda, \beta_2]$ with $1/3 < \lambda < 2/3$: medium special-$q$;
  \item For $q = L_Q[1/3, \beta_3]$: small special-$q$.
\end{enumerate}
The proof of the complexity is not trivial at all, and since this step
is allegedly cheaper than the two main phases of sieving and linear algebra,
whose complexity is $L_Q[1/3, (\frac{64}{9})^{1/3}]$, the proofs are
not always expanded.

There is a detailed proof in \cite[\S 4.3]{PKC:ComSem06} and \cite[\S 7.3]{Bar13thesis} for prime
fields $\FF_p$. We found another detailed proof in \cite[\S
B]{IMA:JLNT09} for large characteristic fields $\FF_{p^n}$, however
this was done for the polynomial selection of \cite[\S 3.2]{C:JLSV06}
(which has the same main asymptotic complexity $L_Q[1/3, (\frac{64}{9})^{1/3}]$).
In \cite[\S 4, pp.~144--150]{Mat06} the NFS-DL algorithm is not
proposed in the same order: the booting and descent steps 
(step (5) of the algorithm in \cite[\S 2]{Mat06})
are done as a first sieving, then the relations are added to the matrix
that is solved in the linear algebra phase.
What corresponds to a booting step is proved to have a complexity
bounded by $L_Q[1/3, 3^{1/3}]$ and there is a proof that the descent
phase has a smaller complexity than the booting step. 
There is a proof for the JLSV\sub{1} polynomial selection in 
\cite[\S C]{IMA:JLNT09} and \cite[\S A]{BarPie2014} for a MNFS
variant.
We summarize in Tab.~\ref{tab:cpx-descent-bib} the asymptotic
complexity formulas for the booting step and the descent step that
we found in the available papers.

\begin{table}[htbp]
\centering
\caption{Complexity of the booting step and the descent step for
 computing one individual DL, in $\FF_p$ and
 $\FF_{p^n}$, in medium and large characteristic. The complexity is
 given by the formula $L_Q[1/3, c]$, only the constant $c$ is given in
the table for ease of notation. The descent of a medium special-$q$,
bounded by $L_Q[\lambda, c]$ with $1/3 < \lambda < 2/3$, is proven to
be negligible compared to the large and small special-$q$ descents. In
\cite[\S B,C]{IMA:JLNT09}, the authors used a sieving technique over ideals of
degree $t > 1$ for large and medium special-$q$ descent.}
\label{tab:cpx-descent-bib}
\begin{tabular}{|c|c|c||c||c|c|c|c|}
\hline reference                & finite field & polynomial        & target & booting & \mc{3}{c|}{descent step} \\
\cline{6-8}                     &              &  selection        & norm bound &step    & large&med.&small\\
\hline \cite[\S 4.3]{PKC:ComSem06}  & $\FF_p$  & JL03 \cite{JoLe03} & $p$ & 1.44 &  \mc{3}{c|}{$<$1.44} \\
\hline \cite[Tab.~7.1]{Bar13thesis} & $\FF_p$  & JL03 \cite{JoLe03} & $p$ & 1.23 & 1.21 & neg. & 0.97 \\
\hline \cite[\S 4]{Mat06}       & $\FF_{p^n}$, large $p$ & gJL       & $Q$ & 1.44 & \mc{3}{c|}{$<$ 1.44} \\
\hline \cite [\S B]{IMA:JLNT09} & $\FF_{p^n}$, large $p$ & JLSV\sub{2}  & $Q$ & 1.44 & -- & neg. & 1.27 \\
\hline \cite [\S C]{IMA:JLNT09} & $\FF_{p^n}$, med. $p$  & JLSV\sub{1}
variant& $Q^{1+\alpha}$, $\alpha \simeq0.4$ & 1.62 &-- & neg. & 0.93 \\
\hline \cite[\S A]{BarPie2014}  & $\FF_{p^n}$, med. $p$  & JLSV\sub{1} & $Q^{3/2}$ & 1.65 & \mc{3}{c|}{$\leq 1.03$} \\
\hline
\end{tabular}
\end{table}

Usually, the norm of the target is assumed to be bounded by $Q$ (this
is clearly the case for prime fields $\FF_p$). The resulting
initialization step (finding a boot for the descent) has complexity
$L_Q[1/3, 3^{1/3} \approx 1.44]$. 
Since the
large special-$q$ descent complexity depends on the size of the largest special-$q$ of
the boot, lowering the norm, hence the booting step complexity
\emph{and} the largest special-$q$ of the boot also
decrease the large special-$q$ descent step complexity. 
It would be a considerable project to rewrite new proofs for each polynomial
selection method, according to the new booting step
complexities. However, its seems to us that by construction, the large
special-$q$ descent
step in  these cases has a (from much to slightly) smaller complexity than the
booting step.
The medium special-$q$ descent step has a negligible
cost in the cases considered above. Finally, the small special-$q$ descent step does not depend on
the size of the boot but on the polynomial properties (degree, and
coefficient size). 
We note that for the JLSV\sub{2} polynomial
selection, the constant of the complexity is 1.27. It would be
interesting to know the constant for the gJL and Conjugation methods. 

The third and final step of individual logarithm computation is very
fast. It 
combines all of the logarithms computed before, to get the final discrete
logarithm of the target.

\section{Computing a Preimage in the Number Field}
\label{sec: preimage}

Our main idea is to compute a preimage in the number field with smaller
degree (less than $\deg s$) and/or of coefficients of
reduced size, by using the subfield structure of $\FF_{p^n}$. 
We at least have one non-trivial subfield: $\FF_p$. 
In this section, we reduce the size of the coefficients of the
preimage. This reduces its norm and give the first part of the proof of 
Theorem~\ref{th: smaller norm bound}. In the following section, we
will reduce the degree of the preimage when $n$ is even, completing
the proof.

\begin{lemma} \label{lemma:log equality up to subgroup elt}
Let $s \in \FF_{p^n}^{*} = \sum_{i=0}^{\deg s} s_i x^i$, with $\deg s <
n$. Let $\ell$ be a non-trivial divisor of $\Phi_n(p)$. 
Let $s' = u \cdot s$ with $u$ in a proper subfield of $\FF_{p^n}$. Then
\begin{equation}
\log s' \equiv \log s \bmod \ell ~.
\end{equation}
\end{lemma}

\begin{proof}
We start with $\log s' = \log s + \log u $
and since $u$ is in a proper subfield, we have $u^{(p^n-1)/\Phi_n(p)} = 1$, 
then $u^{(p^n-1)/\ell} = 1$. Hence the logarithm of $u$ modulo $\ell$ is zero,
and $\log s' \equiv \log s \bmod \ell$.
\qed
\end{proof}

\begin{example}[Monic preimage]
Let $s'$ be equal to $s$ divided by its
leading term, $s' = \frac{1}{s_{\deg s}} s\in \FF_{p^n}$.
We have $\log s' \equiv \log s \bmod \ell$.
\end{example}
We assume in the following that the target $s$ is monic since dividing
by its leading term does not change its logarithm modulo $\ell$.

\subsection{Preimage Computation in the JLSV\sub{1} Case}
\label{subsec: LLL monic preimage JLSV1}
Let $s = \sum _{i=0}^{n-1} s_i x^i \in \FF_{p^n}$ with  $s_{n-1} = 1$.
We define a lattice of dimension $n$ by the $n \times n$ matrix
$$ L = 
\begin{bmatrix}
p         &                  &           &       \\
          & \ddots           &           &       \\
          &                  & p         &       \\
\lift{s}_0 &\ldots&\lift{s}_{n-2}& \pM 1 \\
\end{bmatrix}
\underset{n\times n~~~~}{
  \begin{matrix}
    \sss{0} \\ \sss{\vdots} \\ \sss{n-2} \\ \sss{n-1} \\
  \end{matrix}
}
\begin{array}{l}
   \left \rbrace \begin{array}{l}
       \phantom{I} \\
       n-1 \mbox{ rows} \\
       \phantom{\ddots} \\
   \end{array}\right. \\
   ~ \} \mbox{row } n-1 \mbox{ with } \lift{s} \mbox{ coeffs}  \\
\end{array} $$
with $p$ on the diagonal for the first $n-1$ rows (from 0 to $n-2$), and
the coefficients of the monic element $\lift{s}$ on row $n-1$.
Applying the LLL algorithm to $M$, we obtain a reduced element
$\lift{r} = \sum_{i=0}^{n-1} \lift{r}_i X^i \in K_f$ such that 
$$ \lift{r} = \sum_{i=0}^{n-1} a_i L_i$$
with $L_i$ the vector defined by the $i$-th row of the matrix and $a_i$ a
scalar in $\ZZ$. We map this equality in $\FF_{p^n}$ with $\rho$. All
the terms cancel out modulo $p$ except the line with $\lift{s}$:
$$ \rho(\lift{r}) \equiv \rho(a_{n-1}) \cdot \rho(\lift{s}) = u \cdot s
\mod (p, \psi) $$
with $u = \rho(a_{n-1}) \in \FF_p$. Hence, by 
Lemma~\ref{lemma:log equality up to subgroup elt},
\begin{equation}
\log \rho(\lift{r}) \equiv \log s \bmod \ell~.
\end{equation}
Moreover,
$$ \|\lift{r}\|_\infty \leq C p^{(n-1)/n} ~.$$
It is straightforward, using Inequality \eqref{eq:norm bound}, to
deduce that 
$$\Norm_{K_f/\QQ}(\lift{r}) = O\bigl(p^{\frac{3}{2}(n-1)}\bigr)  = 
O\bigl(Q^{\frac{3}{2} - \frac{3}{2n}}\bigr)~.$$
We note that this first simple improvement applied to the JLSV\sub{1}
construction is already better than
doing nothing: in that case, $\Norm_{K_f/\QQ}(s) = O(Q^{\frac{3}{2} - \frac{1}{2n}})$.
The norm of $\lift{r}$ is smaller by a factor of size $Q^{\frac{1}{n}}$.
For $n=2$ we have $\Norm_{K_f/\QQ}(\lift{r}) = O(Q^{\frac{3}{4}})$ but
for $n=3$, the bound is $\Norm_{K_f/\QQ}(\lift{r}) = O(Q)$, and for
$n=4$, $O(Q^{11/8})$. This is already too large. We would
like to obtain such a bound, strictly smaller than $O(Q)$, for any $n$.

\subsection{Preimage Computation in the gJL and Conjugation Cases}
\label{subsec: LLL monic preimage gJL CONJ}
Let $s = \sum _{i=0}^{n-1} s_i x^i \in \FF_{p^n}$ with $s_{n-1} = 1$.
In order to present a generic method for both the gJL and the Conjugation
methods, we denote by $d_f$ the degree of $f$. In the gJL case we have
$d_f = d+1 \geq n+1$, while in the Conjugation case, $d_f = 2n$.
We define the $d_f \times d_f$ matrix 
with $p$ on the diagonal for the first $n-1$ rows,
and the coefficients of the monic element $\lift{s}$ on row $n-1$. 
The rows $n$ to $d_f$ are filled with the
coefficients of the monic polynomial $x^j\psi$, with $0 \leq j \leq
d_f-n$.

$$ L = 
\begin{bmatrix}
p          &           &               &           &      &  & \\
           & \ddots    &               &           &      &  & \\
           &           & p             &           &      &  & \\
\lift{s}_0  & \ldots    & \lift{s}_{n-2} & 1         &      &  & \\
\psi_{0}&\psi_{1}& \cdots    & \psi_{n-1} & 1    &  & \\
           & \ddots    & \ddots        &           &\ddots&\ddots        &  \\
           &           &\psi_{0}    &\psi_{1}&\cdots&\psi_{n-1} &1 \\
\end{bmatrix}
\underset{d_f \times d_f~}{
  \begin{matrix}
    \sss{0} \\ \sss{\vdots} \\ \sss{n-2} \\ \sss{n-1} \\ \sss{n} \\ \sss{\vdots} \\ \sss{d_f-1} \\ 
  \end{matrix}
}
\begin{array}{l}
   \left \rbrace \begin{array}{l}
       \phantom{I} \\
       n-1 \mbox{ rows } \\
       \phantom{\ddots} \\
   \end{array}\right. \\
   ~ \} \mbox{row } n-1 \mbox{ with } \lift{s} \mbox{ coeffs}  \\
   \left \rbrace \begin{array}{l}
         \phantom{I} \\
         d_f-n  \mbox{ rows } \mbox{with } \psi \mbox{ coeffs} \\
        \phantom{\ddots} \\
   \end{array}\right.
\end{array}
$$

Applying the LLL algorithm to $L$, we obtain a reduced element
$\lift{r} = \sum_{i=0}^{d_f-1} \lift{r}_i X^i \in K_f$ such that 
$ \lift{r} = \sum_{i=0}^{d_f-1} a_i L_i$
where $L_i$ is the $i$-th row vector of $L$ and $a_i$ is a
scalar in $\ZZ$. We map this equality into $\FF_{p^n}$ with $\rho$. All
the terms cancel out modulo $(p, \psi)$ except the one with $\lift{s}$ coefficients:
$$ \rho(\lift{r}) \equiv \rho(a_{n-1}) \cdot \rho(\lift{s}) = u \cdot s
\mod (p, \psi) $$
with $u = \rho(a_{n-1}) \in \FF_p$. Hence, by 
Lemma~\ref{lemma:log equality up to subgroup elt},
\begin{equation}
\log \rho(\lift{r}) \equiv \log s \bmod \ell~.
\end{equation}
Moreover,
$$ \|\lift{r}\|_\infty \leq C p^{(n-1)/d_f}~.$$
It is straightforward, using
Inequality  \eqref{eq:norm bound}, to deduce that 
$$\Norm_{K_f/\QQ}(\lift{r}) = O\bigl(p^{n-1}\bigr)  = O\bigl(Q^{1 - 1/n}\bigr)~.$$
Here we obtain a bound that is \emph{always strictly smaller} than $Q$ for any $n$. In
the next section we show how to improve this bound to $O\bigl(Q^{1 -
  2/n}\bigr)$ when $n$ is even
and the number field defined by $\psi$ has a well-suited quadratic subfield.

\section{Preimages of Smaller Norm  with Quadratic Subfields}
\label{sec: preimage smaller norm}
Reducing the degree of $s$ can reduce the norm size in the number
field for the JLSV\sub{1} polynomial construction. We present a way to
compute $r\in \FF_{p^n}$ of degree $n-2$ from $s \in \FF_{p^n}$ of
degree $n$ in the given representation of $\FF_{p^n}$, and $r, s$
satisfying Lemma~\ref{lemma:log equality up to subgroup elt}. 
We need $n$ to be even and the finite field $\FF_{p^n}$ to be expressed
as a degree-$n/2$ extension of a quadratic extension defined by a
polynomial of a certain form.
We can define another lattice with $r$ and get a preimage of degree
$n-2$ instead of $n-1$ in the number field.
This can be interesting with the JLSV\sub{1} method. 
Combining this method with the previous one of Sec.~\ref{sec: preimage} leads to
our proof of Theorem~\ref{th: smaller norm bound}.

\subsection{Smaller Preimage Degree}
In this section, we prove that when $n$ is even and $\FF_{p^n} =
\FF_p[X]/(\psi(X))$ has a quadratic base field $\FF_{p^2}$ of a
certain form, from a random element $s \in \FF_{p^n}$ with $s_{n-1}
\neq 0$, we can compute an element $r \in \FF_{p^n}$ with $r_{n-1} = 0$, and
$s = u \cdot r$ with $u \in \FF_{p^2}$. Then, using 
Lemma~\ref{lemma:log equality up to subgroup elt},
we will conclude that $\log r \equiv \log s \bmod \ell$.

\begin{lemma}\label{lemma:Fp2 simplification}
Let $\psi(X)$ be a monic irreducible polynomial of $\FF_p[X]$ of
even degree $n$ with a
quadratic subfield defined by the polynomial $P_y = Y^2 + y_1 Y + y_0$. 
Moreover, assume that $\psi$  splits over $\FF_{p^2} = \FF_p[Y]/(P_y(Y))$ as
$$\begin{array}{l l}  &  \psi(X) = (P_z(X) - Y)(P_z(X) - Y^p) \\ 
\mbox{ or }           &  \psi(X) = (P_z(X) - Y X)(P_z(X) - Y^p X)  \\
\end{array}$$
with $P_z$ monic, of degree $n/2$ and coefficients in $\FF_p$.
Let $s \in \FF_p[X]/(\psi(X))$ a random element, $s =
\sum_{i=0}^{n-1} s_i X^i$.

Then there exists $r \in \FF_{p^n}$ 
monic and of degree $n-2$ in $X$, and  $u \in \FF_{p^{2}}$, 
such that $s = u \cdot r$ in $\FF_{p^n}$.
\end{lemma}

We first give an example for $s \in \FF_{p^4}$ then present a
constructive proof.

\begin{example}
Let $P_y = Y^2 + y_1 Y + y_0$ be a monic irreducible polynomial over
$\FF_{p}$ and set $\FF_{p^2} = \FF_p[Y]/(P_y(Y))$. 
Assume that $Z^2 - YZ + 1$ is irreducible over $\FF_{p^2}$ and set
$\FF_{p^4} = \FF_{p^2}[Z] /(Z^2 - YZ + 1)$. 
Let $\psi = X^4 + y_1 X^3 + (y_0 + 2) X^2 + y_1 X + 1$ be a monic
reciprocal polynomial. By construction,  $\psi$
factors over $\FF_{p^2}$ into $(X^2 - YX + 1)(X^2 - Y^p X + 1)$
and $\FF_p[X]/(\psi(X))$ defines a
quartic extension $\FF_{p^4}$ of $\FF_p$.
We have these two representations for $\FF_{p^4}$:

$$  \begin{array}{cc l}
    \FF_{p^4} &=& \FF_{p^2}[Z]/(Z^2 - YZ + 1) \\
    \mid      & & \\
    \FF_{p^2} &=& \FF_p[Y]/(Y^2 + y_1 Y + y_0) \\
    \mid      & & \\
    \FF_p     & & 
  \end{array} 
\begin{array}{c} \mbox{ and } \\ \\ \\ \\ \\ \end{array}
 \begin{array}{cc l}
    \FF_{p^4} &=& \FF_p[X]/(X^4+y_1X^3+(y_0+2)X^2+y_1X+1) \\
    \mid  & & \\
    \mid  & & \\
    \mid  & & \\
    \FF_p & & \\
\end{array}$$
 \end{example}
$ $
\begin{proof}[of Lemma~\ref{lemma:Fp2 simplification}]
Two possible extension field towers are:
$$
\begin{array} {c c c }
\begin{array}{l}
\FF_{p^n} = \FF_{p^2}[Z]/(P_z(Z) - Y) \\
\mid \\
\FF_{p^2} = \FF_p[Y]/(P_y(Y)) \\
\mid \\
\FF_p \\
\end{array} & \mbox{ and } &
\begin{array}{l}
\FF_{p^n} = \FF_{p^2}[Z]/(P_z(Z) - YZ) \\
\mid \\
\FF_{p^2} = \FF_p[Y]/(P_y(Y)) \\
\mid \\
\FF_p \\
\end{array}\\
\end{array}
$$

We write $s$ in the following representation to emphasize the subfield structure:
$$ s = \sum_{i=0}^{n/2-1} (a_{i0} + a_{i1} Y) Z^i \mbox{ with } a_{ij} \in \FF_p~. $$

\begin{enumerate}
\item If $\psi = P_z(Z) - Y$ then we can divide $s$ by $u_{LT} = a_{n/2,0} +
  a_{n/2,1} Y  \in \FF_{p^2}$ (the leading term in $Z$, i.e. the coefficient of
  $Z^{n/2}$) to make $s$ monic in $Z$ up to a subfield
  cofactor $u_{LT}$:
$$ \frac{s}{u_{LT}} = 
       \sum_{i=0}^{n/2 - 2} (b_{i0} + b_{i1} Y) Z^i ~~~ + Z^{n/2-1} \ , $$
with the coefficients $b_{ij}$ in the base field $\FF_p$, and $b_{i0} +
b_{i1} Y = (a_{i0} + a_{i1} Y) / u_{LT}$. 
Since $P_z(Z) = Y$ and $Z = X$ in $\FF_{p^n}$
by construction, we replace $Y$ by $P_z(Z)$ and $Z$ by $X$
to get an expression for $s$ in $X$:
$$ \frac{s}{u_{LT}} = \sum_{i=0}^{n/2 - 2} (b_{i0} +
b_{i1} P_z(X)) X^i + X^{n/2-1} = r(X)~.$$
The degree in $X$ of $r$ is $\deg r = \deg P_z(X) X^{n/2-2} = n-2 $  instead of $\deg s = n-1$.
We set $u = 1/u_{LT}$. By construction, $u \in \FF_{p^2}$. 
We conclude that $s = u r \in \FF_{p^n}$, with $\deg r =
n-2$ and $u \in \FF_{p^2}$.
\item If $\psi = P_z(Z) - YZ$ then we can divide $s$ by
  $u_{CT} = a_{00} + a_{01} Y \in \FF_{p^2}$
(the constant term in $Z$) to make the constant coefficient of $s$ to
be 1:
$$ \frac{s}{u_{CT}} = 1 + \sum_{i=1}^{n/2-1} (b_{i0} + b_{i1} Y) Z^i $$
with $b_{ij} \in \FF_p$. 
Since $P_z(Z) = YZ$ and $Z = X$ in $\FF_{p^n}$
by construction, 
we replace $YZ$ by $P_z(Z)$ and $Z$ by $X$ to get
$$ \frac{s}{u_{CT}} = 1 + \sum_{i=1}^{n/2 - 1} (b_{i0}
X^i + b_{i1} P_z(X)  X^{i-1}) = r(X)~.$$

The degree in $X$ of $r$ is $\deg r = \deg P_z(X) X^{n/2-1-1} = n-2 $
instead of $\deg s = n-1$. 
We set $u = 1/u_{CT} $.  By construction, $u \in \FF_{p^2}$. 
We conclude that $s = u r \in \FF_{p^n}$, with $\deg r =
n-2$ and $u \in \FF_{p^2}$.
\qed
\end{enumerate}
\end{proof}

Now we apply the technique described in 
Sec.~\ref{subsec: LLL monic preimage JLSV1} to reduce the
coefficient size of $r$ in the JLSV\sub{1} construction.
We have $r_{n-1} = 0$ and we assume that $r_{n-2} = 1$. 
We define the lattice by the $(n-1)\times(n-1)$ matrix
$$ L = 
\begin{bmatrix}
p         &                  &           &       \\
          & \ddots           &           &       \\
          &                  & p         &       \\
\lift{r}_0 &\ldots&\lift{r}_{n-3}& \pM 1 \\
\end{bmatrix}
\underset{n-1\times n-1~~~~}{
  \begin{matrix}
    \sss{0} \\ \sss{\vdots} \\ \sss{n-3} \\ \sss{n-2} \\
  \end{matrix}
}
\begin{array}{l}
   \left \rbrace \begin{array}{l}
       \phantom{I} \\
       n-2 \mbox{ rows} \\
       \phantom{\ddots} \\
   \end{array}\right. \\
   ~ \} \mbox{row } n-2 \mbox{ with } \lift{r} \mbox{ coeffs}  \\
\end{array} $$

After reducing the lattice with LLL, we obtain an element
$\lift{r}'$ whose coefficients are bounded by $C
p^{\frac{n-2}{n-1}}$. The norm of
$\lift{r}'$ in the number field $K_f$ constructed with the
JLSV\sub{1} method is 
$$\Norm_{K_f/\QQ}(\lift{r}') = O(p^{\frac{3}{2} n - 2 -
  \frac{1}{n-1}} )  = O(Q^{\frac{3}{2} - \frac{2}{n} - \frac{1}{n(n-1)}}) ~.$$
This is better than the previous $O\bigl(Q^{\frac{3}{2} -
  \frac{3}{2n}}\bigr)$ case: the norm is smaller by a factor of size 
$O\bigl(Q^{\frac{1}{2}n + \frac{1}{n(n-1)}}\bigr)$. 
For $n = 4$, we obtain $\Norm_{K_f/\QQ}(\lift{r}') =
O\bigl(Q^{\frac{11}{12}}\bigr)$, which is strictly less than $O(Q)$. 

We can do even better by re-using the element $\lift{r}$ of degree
$n-2$ and the given one $s$ of degree $n-1$, and combining them.

\paragraph{Generalization to subfields of higher degrees.}
It was pointed out to us by an anonymous reviewer that more generally,
by standard linear algebra arguments, 
for $m \mid n$ and $s \in \FF_{p^n}$, there exists a non-zero $u \in
\FF_{p^m}$ such that $s \cdot u$ is a polynomial of degree at most $n-m$.

\subsection{Smaller Preimage Norm}
\label{subsubsec:smaller norm with deg 2 subfield}

First, suppose that the target element $s = \sum_{i=0}^{n-1} s_i x^i$ satisfies
$s_{n-1} = 0$ and $s_{n-2} = 1$. 
We can define a lattice whose vectors,
once mapped to $\FF_{p^n}$, are either 0 (so vectors are sums of multiples of
$p$ and $\psi$) or are multiples of the initial target $s$, satisfying 
Lemma~\ref{lemma:log equality up to subgroup elt}. The above $r$ of
degree $n-2$ is a good candidate. The initial $s$ also. If there is no
initial $s$ of degree $n-1$, then simply take at random any $u$ in a
proper subfield of $\FF_{p^n}$ which is not $\FF_p$ itself and set $s
= u \cdot r$. Then $s$ will have $s_{n-1} \neq 0$. 
Then define the lattice
$$ L = 
\begin{bmatrix}
p         &          &             &       &   \\
          & \ddots   &             &       &   \\
          &          & p           &       &   \\
\lift{r}_0 &\ldots    &\lift{r}_{n-3}& \pM 1 &   \\
\lift{s}_0 &\ldots    &\lift{s}_{n-3}&\lift{s}_{n-2}& \pM 1 \\
\end{bmatrix}
\underset{n\times n~~~}{
  \begin{matrix}
    \sss{0} \\ \sss{\vdots} \\ \sss{n-3} \\ \sss{n-2} \\ \sss{n-1} \\
  \end{matrix}
}
\begin{array}{l}
   \left \rbrace \begin{array}{l}
       \phantom{I} \\
       n-2 \mbox{ rows} \\
       \phantom{\ddots} \\
   \end{array}\right. \\
   ~ \} \mbox{ row } n-2 \mbox{ with } \lift{r} \mbox{ coeffs}  \\
   ~ \} \mbox{ row } n-1 \mbox{ with } \lift{s} \mbox{ coeffs}  \\
\end{array} $$
and use it in place of the lattices of Sec.~\ref{subsec: LLL monic preimage JLSV1}
or \ref{subsec: LLL monic preimage gJL CONJ}.
\subsection{Summary of results}
We give in Table~\ref{tab: new booting step cpx} the previous and new upper bounds for the norm of
$s$ in a number field $K_f$ for three
polynomial selection methods: the JLSV\sub{1} method, the generalized
Joux--Lercier method and the Conjugation method, and the complexity of
the booting step to find a ${B}$-smooth decomposition of $\Norm_{K_f/\QQ}(s)$.
We give our practical results for small $n$, where there are 
the most dramatic improvements. 
We obtain the optimal norm size of $Q^{\varphi(n)/n}$ for
$n=2,3,5$ with the gJL method and also for $n=4$ with the Conjugation method.

\begin{table}[hbtp]
\centering
\caption{Norm bound of the preimage with our method, and booting step complexity.}
\label{tab: new booting step cpx}
\hspace*{-0.05\textwidth}%
\begin{tabular}{|c|c||c|c|c||c|c||c|c|c|c|c|}
\hline \mr{2}{$\FF_{p^n}$}& poly. & \mc{3}{c|}{norm bound}                            & \mc{2}{c|}{booting step $L_Q[\frac{1}{3},c]$} & \mc{5}{c|}{practical values of $c$} \\
\cline{3-12}                 & selec.  & nothing & JLSV & this work & prev & this work & $n=2$ & $n=3$ & $n=4$ & $n=5$ & $n=6$ \\
\hline any $n>1$    & \mr{2}{gJL} & \mr{2}{$Q^{1+\frac{1}{n}}$} & \mr{2}{$Q$} & $Q^{1-1/n}$ & \mr{2}{1.44} & $(3(1-\frac{1}{n}))^{1/3}$ &1.14 & 1.26 & --     & 1.34 &  --    \\
     even $n\geq 4$ &             &                   &           & $Q^{1-2/n}$ &              & $(3(1-\frac{2}{n}))^{1/3}$ & --    & --     & 1.14 & --     & 1.26 \\
\hline any $n>1$    & \mr{2}{Conj} & \mr{2}{$Q^{2}$} & \mr{2}{$Q$} & $Q^{1-1/n}$ & \mr{2}{1.44} & $(3(1-\frac{1}{n}))^{1/3}$ & 1.14 & 1.26 & --     & 1.34 &  --    \\
     even $n\geq 4$ &              &                &             & $Q^{1-2/n}$ &              & $(3(1-\frac{2}{n}))^{1/3}$ & --   &   --   & 1.14 &  --    & 1.26 \\
\hline any $n>1$    & \mr{2}{JLSV\sub{1}} & \mr{2}{$Q^{\frac{3}{2}-\frac{1}{2n}}$} & \mr{2}{$Q^2$} & $Q^{3/2-3/(2n)}$ & \mr{2}{1.65} & $(\frac{9}{2}(1-\frac{1}{n}))^{1/3}$ & 1.31 & 1.44 & --    & 1.53 &  --     \\
     even $n\geq 4$ &                  &                  &            & $Q^{3/2-5/(2n)}$ &              & $(\frac{3}{2}(3-\frac{5}{n}))^{1/3}$ & --    &   --   & 1.38 &  --    & 1.48 \\
\hline
\end{tabular}
\end{table}

\section{Practical examples}
\label{sec: practical examples}
We present an example for each of the three polynomial selection
methods we decided to study. The Conjugation method provides the
best timings for $\FF_{p^2}$ at 180 dd \cite{EC:BGGM15}. We apply the gJL
method to $\FF_{p^3}$ according to \cite[Fig.~3]{EC:BGGM15}. We decided
to use the JLSV\sub{1} method for $\FF_{p^4}$ \cite[Fig.~4]{EC:BGGM15}.

\subsection{Examples for Small $n$  and $p^n$ of 180 Decimal Digits (dd)}
\label{subsec: 180dd examples}
\subsubsection{Example for $n=2$, Conjugation Method.}
We take the parameters of the record in \cite{EC:BGGM15}: 
$p$ is a 90 decimal digit (300 bit) prime number, and
$f, \psi$ are computed with the Conjugation method.
We choose a target $s$ from the decimal digits of $\exp(1)$.
$$\begin{array}{lcl}
p&=&\coeffss{314159265358979323846264338327950288419716939937510582097494459230781640628620899877709223} \\
f&=& x^4+1 \\
\psi &=& x^2 +\coeffss{107781513095823018666989883102244394809412297643895349097410632508049455376698784691699593}~ x + 1 \\
s &=& \coeffss{271828182845904523536028747135319858432320810108854154561922281807332337576949857498874314}~x \\
  & & + \coeffss{95888066250767326321142016575753199022772235411526548684808440973949208471194724618090692} \\
\end{array}$$
We first compute $s' = \frac{1}{s_0}s$ then reduce
$$L = \begin{bmatrix}
p         & 0         & 0         & 0 \\
s_0'      & 1         & 0         & 0 \\
1         & \psi_1 & 1         & 0 \\
0         & 1         & \psi_1 & 1 \\
\end{bmatrix} $$
then $\LLL(L)$ produces $\lift{r}$ of degree 3 and coefficient size
$O(p^{1/4})$. Actually LLL outputs four short vectors, hence we get
four small candidates for $\lift{r}$, each of norm
 $\Norm_{K_f/\QQ}(\lift{r}) = O(p) = O(Q^{1/2}) =
O(Q^{\varphi(n)/n})$, i.e. 90 dd. 
To slightly improve the smoothness search time, we can compute linear
combinations of these four reduced preimages.
$$\begin{array}{l}
\coeffss{3603397286457205828471} x^3 + \coeffss{13679035553643009711078} x^2 + 
\coeffss{5577462470851948956594} x + \coeffss{856176942703613067714} \\
\coeffss{9219461324482190814893}x^3 - \coeffss{4498175796333854926013}x^2 + 
\coeffss{8957750025494673822198}x + \coeffss{1117888241691130060409} \\
\coeffss{28268390944624183141702}x^3 + \coeffss{5699666741226225385259}x^2 - 
\coeffss{17801940403216866332911}x + \coeffss{5448432247710482696848} \\
\coeffss{3352162792941463140060}x^3 + \coeffss{3212585012235692902287}x^2 - 
\coeffss{5570636518084759125513}x + \coeffss{46926508290544662542327} \\
\end{array}$$
The norm of the first element is 
$$\Norm_{K_f/\QQ}(\lift{r}) = \coeffss{21398828029520168611169045280302428434866966657097075761337598070760485340948677800162921} $$
of 90 decimal digits, as expected.
For a close to optimal running-time of $L_Q[1/3, 1.14] \sim 2^{40}$ to find a boot, the special-$q$
bound would be around 64 bits.

\subsubsection{Example for $n=3$, gJL Method.}
We take $p$ of 60 dd (200 bits) so that $\FF_{p^3}$ has size 180 dd
(600 bits) as above. We took $p$ a prime made of the 60 first
decimal digits of $\pi$. We constructed $f, \psi, g$ with the gJL
method described in \cite{EC:BGGM15}.  

$$\begin{array}{lcl}
p &=&\coeffss{314159265358979323846264338327950288419716939937510582723487} \\
f &=& x^4 - x + 1 \\
\psi &=& x^3 +\coeffss{227138144243642333129902287795664772043667053260089299478579} x^2 \\
        & & + \coeffss{126798022201426805402186761110440110121157863791585328913565} x 
            + \coeffss{86398309157441443539791899517788388184853963071847115552638} \\
g &=& \coeffss{2877670889871354566080333172463852249908214391} x^3 + 
\coeffss{6099516524325575060821841620140470618863403881} x^2 \\
  & & - \coeffss{10123533234834473316053289623165756437267298403} x
  + \coeffss{2029073371791914965976041284208208450267120556} \\
s &=& \coeffss{271828182845904523536028747135319858432320810108854154561922}x^2 
      + \coeffss{281807332337576949857498874314095888066250767326321142016575} x \\
  & & + \coeffss{75319902277223541152654868480858951626493739297259139859875} \\
\end{array}$$
We set $s' = \frac{1}{s_2}s$. The lattice to be reduced is
$$L = \left [ \begin{array}{cccc}
p & 0 & 0 & 0 \\
0 & p & 0 & 0 \\
s_0' & s_1' & 1 & 0 \\
\psi_0 & \psi_1 & \psi_2 & 1 \\
\end{array}  \right]  $$
then $\LLL(L)$ computes four short vectors $\lift{r}$ of degree
3, of coefficient size $O(p^{1/2})$, and of norm size $\Norm_{K_f/\QQ}(\lift{r}) = O(p^2) = O(Q^{2/3}) =
O(Q^{\varphi(n)/n})$.

$$\begin{array}{l}
\coeffss{159774930637505900093909307018} x^3 + \coeffss{165819631832105094449987774814} x^2 
 + \coeffss{177828199322419553601266354904} x - \coeffss{159912786936943488400590389195} \\
\coeffss{136583029354520905232412941048} x^3 - \coeffss{521269847225531188433352927453} x^2 
 + \coeffss{322722415562853671586868492721} x + \coeffss{255238068915917937217884608875} \\
\coeffss{118289007598934068726663000266} x^3 + \coeffss{499013489972894059858543976363} x^2 
 - \coeffss{105084220861844155797015713666} x + \coeffss{535978811382585906107397024241} \\
\coeffss{411603890054539500131474313773} x^3 - \coeffss{240161030577722451131067159670} x^2 
 - \coeffss{373289346204280810310169575030} x - \coeffss{389720783049275894296185820094} \\
\end{array}$$
The norm of the first element is 
$$\Norm_{K_f/\QQ}(\lift{r}) = 
\coeffss{997840136509677868374734441582077227769466501519927620849763845265357390584602475858356409809239812991892769866071779}$$
of 117 decimal digits (with $\frac{2}{3} 180 = 120$ dd).
For a close to optimal running-time of $L_Q[1/3, 1.26] \sim 2^{45}$ to find a boot, the special-$q$
bound would be around 77 bits.

\subsubsection{Example for $n=4$, JLSV\sub{1} Method.}

$$\begin{array}{lcl}
p&=&\coeffss{314159265358979323846264338327950288419980011} \\
\ell &=&
\coeffss{49348022005446793094172454999380755676651143247932834802731698819521755649884772819780061} \\
f = \psi&=& x^4 + x^3 + \coeffss{70898154036220641093162} x^2 + x + 1 \\
g &=& \coeffss{101916096427067171567872} x^4 + \coeffss{101916096427067171567872} x^3 
  +\coeffss{220806328874049898551011} x^2 \\
        & & + \coeffss{101916096427067171567872} x + \coeffss{101916096427067171567872} \\
s &=& \coeffss{271828182845904523536028747135319858432320810}x^3 + \coeffss{108854154561922281807332337576949857498874314} x^2 \\
  & & + \coeffss{95888066250767326321142016575753199022772235} x + \coeffss{41152654868480844097394920847127588391952018} \\
\end{array}$$
We set $s' = \frac{1}{s_3}s$.
The subfield simplification for $s$ gives
$$r = x^2 + \coeffss{134969122397263102979743226915282355400161911} x
+ \coeffss{104642440649937756368545765334741049207121011} \ .$$
We reduce the lattice defined by
$$L = \left [ \begin{array}{ccccc}
p   & 0   & 0   & 0 \\
0   & p   & 0   & 0 \\
r_0 & r_1 & 1   & 0 \\
s_0' & s_1' & s_2' & 1 \\
\end{array}  \right]  $$
then $\LLL(L)$ produces these four short vectors of degree 3, coefficient size
$O(p^{1/2})$, and norm $\Norm_{K_f/\QQ}(\lift{r}') = O(p^{\frac{7}{2}})
= O(Q^{7/8})$ (smaller than $O(Q)$).
$$\begin{array}{l}
\coeffss{5842961997149263751946} x^3 + \coeffss{290736827330861011376} x^2 
 - \coeffss{5618779793817086743792} x + \coeffss{1092494800287557029045} \\
\coeffss{1640842643903161175359} x^3 + \coeffss{15552590269131889589575} x^2 
 - \coeffss{4425488394163838271378} x - \coeffss{5734086421794811858814} \\
\coeffss{6450686906504525374853} x^3 + \coeffss{13768771242650957399419} x^2 
 + \coeffss{10617583944234090880579} x + \coeffss{16261617079167797580912} \\
\coeffss{16929135804139878865391} x^3 + \coeffss{698185571704810258344} x^2 
 + \coeffss{12799300411012246114079} x - \coeffss{22787282698718065284157} \\
\end{array}$$
The norm of the first element is 
$$\begin{array}{ll}
\Norm_{K_f/\QQ}(\lift{r}') = & \coeffss{14521439292172711151668611104133579982787299949310242601944218977645007049527} \backslash \\
 &
 \coeffss{012365602178307413694530274906757675751698466464799004360546745210214642178285} 
\end{array}$$
of 155 decimal digits (with $\frac{7}{8} 180 = 157.5$).
For a close to optimal running-time of $L_Q[1/3, 1.34] \sim 2^{49}$ to find a boot, the special-$q$
bound would be approximately of 92 bits. This is very large however.

\subsection{Experiments: finding boots for $\FF_{p^4}$ of 120 dd}
\label{subsec:Fp4 120dd}
We experimented our booting step method for $\FF_{p^4}$ of 120 dd (400
bits).
Without the quadratic subfield simplification, the randomized target
norm is bounded by $Q^{9/8}$ of 135 dd (450 bits). The largest
special-$q$ in the boot has size $L_Q[2/3, 3/4]$ (25 dd, 82 bits)
according to Lemma \ref{th: asympt cpx} with $e = 9/8$. The
running-time to find one boot would be $L_Q[1/3, 1.5] \sim 2^{44}$.

We apply the quadratic subfield simplification. 
The norm of the randomized target is $Q^{7/8}$ of 105 dd ($\simeq$ 350
bits). We apply theorem~\ref{th: asympt cpx} with $e = 7/8$. The size of the
largest special-$q$ in the boot will be approximately $L_Q[2/3, 0.634]$ which is
21 dd (69 bits). The running-time needed to find one boot with
the special-$q$ of no more than 21 dd is $L_Q[1/3, 1.38] \sim 2^{40}$
(to be compared with the dominating part of NFS-DL of $L_Q[1/3, 1.923] \sim 2^{57}$).
We wrote a magma program to find boots, using GMP-ECM for $q$-smooth tests. We first set a special-$q$
bound of 70 bits and obtained boots in about two CPU hours. We then reduced the
special-$q$ bound to a machine word size (64 bits) and also found boots in
around two CPU hours. We used an Intel Xeon E5-2609 0 at 2.40GHz with 8 cores.

\section{Conclusion}
We have presented a method to improve the booting step of individual
logarithm computation, the final phase of the NFS algorithm. Our
method is very efficient for small $n$, combined with the gJL or
Conjugation methods; it is also usefull for the JLSV\sub{1} method, but with a
slower running-time. For the
moment, the booting step remains the dominating part of the final
individual discrete logarithm. If our method is
improved, then special-$q$ descent might become the new bottleneck in some cases. 
A lot of work remains to be done on final individual logarithm
computations in order to be able to compute one
individual logarithm as fast as was done in the Logjam~\cite{CCS:ABDGGH15} attack,
especially for $n \geq 3$.

\subsubsection*{Acknowledgements.}
  The author thanks the anonymous reviewers for their
  constructive comments and the generalization of Lemma~\ref{lemma:Fp2 simplification}.
  The author is grateful to Pierrick Gaudry, Fran\c{c}ois Morain and
  Ben Smith.

\bibliographystyle{splncs03}

\begin{thebibliography}{10}
\providecommand{\url}[1]{\texttt{#1}}
\providecommand{\urlprefix}{URL }

\bibitem{Adleman79}
Adleman, L.: A subexponential algorithm for the discrete logarithm problem with
  applications to cryptography. In: 20th FOCS. pp. 55--60. {IEEE} Computer
  Society Press (Oct 1979)

\bibitem{CCS:ABDGGH15}
Adrian, D., Bhargavan, K., Durumeric, Z., Gaudry, P., Green, M., Halderman,
  J.A., Heninger, N., Springall, D., Thom{\'e}, E., Valenta, L., VanderSloot,
  B., Wustrow, E., B{\'e}guelin, S.Z., Zimmermann, P.: Imperfect forward
  secrecy: How {Diffie}-{Hellman} fails in practice. In: Ray, I., Li, N.,
  Kruegel:, C. (eds.) ACM CCS 15. pp. 5--17. {ACM} Press (Oct 2015)

\bibitem{BarPie2014}
Barbulescu, R., Pierrot, C.: The multiple number field sieve for medium- and
  high-characteristic finite fields. LMS J. Comput. Math.  17,  230--246 (1
  2014), \url{http://journals.cambridge.org/article\_S1461157014000369}

\bibitem{Bar13thesis}
Barbulescu, R.: Algorithmes de logarithmes discrets dans les corps finis. Ph.D.
  thesis, Universit{\'e} de Lorraine (2013),
  \url{https://tel.archives-ouvertes.fr/tel-00925228}

\bibitem{BGGM14-DL-record}
Barbulescu, R., Gaudry, P., Guillevic, A., Morain, F.: Discrete logarithms in
  {GF}($p^2$) --- 180 digits (2014),
  \url{https://listserv.nodak.edu/cgi-bin/wa.exe?A2=NMBRTHRY;2ddabd4c.1406},
  announcement available at the {NMBRTHRY} archives

\bibitem{EC:BGGM15}
Barbulescu, R., Gaudry, P., Guillevic, A., Morain, F.: Improving {NFS} for the
  discrete logarithm problem in non-prime finite fields. In: Oswald, E.,
  Fischlin, M. (eds.) EUROCRYPT~2015, Part I. {LNCS}, vol. 9056, pp. 129--155.
  Springer, Heidelberg (Apr 2015), \url{http://hal.inria.fr/hal-01112879}

\bibitem{EC:BGJT14}
Barbulescu, R., Gaudry, P., Joux, A., Thom{\'e}, E.: A heuristic
  quasi-polynomial algorithm for discrete logarithm in finite fields of small
  characteristic. In: Nguyen, P.Q., Oswald, E. (eds.) EUROCRYPT~2014. {LNCS},
  vol. 8441, pp. 1--16. Springer, Heidelberg (May 2014)

\bibitem{C:BlaMulVan84}
Blake, I.F., Mullin, R.C., Vanstone, S.A.: Computing logarithms in {GF}($2^n$).
  In: Blakley, G.R., Chaum, D. (eds.) CRYPTO'84. {LNCS}, vol. 196, pp. 73--82.
  Springer, Heidelberg (Aug 1984)

\bibitem{CEP83}
Canfield, E.R., Erd{\"o}s, P., Pomerance, C.: On a problem of {O}ppenheim
  concerning “factorisatio numerorum”. J. Number Theory  17(1),  1--28
  (1983)

\bibitem{Chen13}
Chen, Y.: Réduction de réseau et sécurité concrète du chiffrement
  complètement homomorphe. Ph.D. thesis, Université Paris 7 Denis Diderot
  (2013), \url{http://www.di.ens.fr/~ychen/research/these.pdf}

\bibitem{PKC:ComSem06}
Commeine, A., Semaev, I.: An algorithm to solve the discrete logarithm problem
  with the number field sieve. In: Yung, M., Dodis, Y., Kiayias, A., Malkin, T.
  (eds.) PKC~2006. {LNCS}, vol. 3958, pp. 174--190. Springer, Heidelberg (Apr
  2006)

\bibitem{JC:Coppersmith93}
Coppersmith, D.: Modifications to the number field sieve. Journal of Cryptology
   6(3),  169--180 (1993)

\bibitem{CopOdlSch86}
Coppersmith, D., Odlzyko, A.M., Schroeppel, R.: Discrete logarithms in
  {GF}($p$). Algorithmica  1(1-4),  1--15 (1986),
  \url{http://dx.doi.org/10.1007/BF01840433}

\bibitem{EC:GamNgu08}
Gama, N., Nguyen, P.Q.: Predicting lattice reduction. In: Smart, N.P. (ed.)
  EUROCRYPT~2008. {LNCS}, vol. 4965, pp. 31--51. Springer, Heidelberg (Apr
  2008)

\bibitem{Gordon93}
Gordon, D.M.: Discrete logarithms in {GF}$(p)$ using the number field sieve.
  SIAM J. Discrete Math  6,  124--138 (1993)

\bibitem{HAKT13}
Hayasaka, K., Aoki, K., Kobayashi, T., Takagi, T.: An experiment of number
  field sieve for discrete logarithm problem over {GF}($p^{12}$). In: Fischlin,
  M., Katzenbeisser, S. (eds.) Number Theory and Cryptography, {LNCS}, vol.
  8260, pp. 108--120. Springer (2013),
  \url{http://dx.doi.org/10.1007/978-3-642-42001-6_8}

\bibitem{JoLe03}
Joux, A., Lercier, R.: Improvements to the general number field for discrete
  logarithms in prime fields. Math. Comp.  72(242),  953--967 (2003)

\bibitem{IMA:JLNT09}
Joux, A., Lercier, R., Naccache, D., Thom{\'e}, E.: Oracle-assisted static
  {Diffie}-{Hellman} is easier than discrete logarithms. In: Parker, M.G. (ed.)
  12th IMA International Conference on Cryptography and Coding. {LNCS}, vol.
  5921, pp. 351--367. Springer, Heidelberg (Dec 2009)

\bibitem{C:JLSV06}
Joux, A., Lercier, R., Smart, N., Vercauteren, F.: The number field sieve in
  the medium prime case. In: Dwork, C. (ed.) CRYPTO~2006. {LNCS}, vol. 4117,
  pp. 326--344. Springer, Heidelberg (Aug 2006)

\bibitem{Kal97}
Kalkbrener, M.: An upper bound on the number of monomials in determinants of
  sparse matrices with symbolic entries. Mathematica Pannonica  73, ~82 (1997)

\bibitem{LLL82}
Lenstra, A., Lenstra, H.W., J., Lovász, L.: Factoring polynomials with
  rational coefficients. Mathematische Annalen  261(4),  515--534 (1982),
  \url{http://dx.doi.org/10.1007/BF01457454}

\bibitem{Mat06}
Matyukhin, D.: Effective version of the number field sieve for discrete
  logarithms in the field {GF}$(p^k)$ (in {{R}}ussian). Trudy po Discretnoi
  Matematike  9,  121--151 (2006), \url{
  http://m.mathnet.ru/php/archive.phtml?wshow=paper&jrnid=tdm&paperid=144&option_lang=eng}

\bibitem{EC:Pierrot15}
Pierrot, C.: The multiple number field sieve with conjugation and generalized
  joux-lercier methods. In: Oswald, E., Fischlin, M. (eds.) EUROCRYPT~2015,
  Part I. {LNCS}, vol. 9056, pp. 156--170. Springer, Heidelberg (Apr 2015)

\bibitem{EC:Weber98}
Weber, D.: Computing discrete logarithms with quadratic number rings. In:
  Nyberg, K. (ed.) EUROCRYPT'98. {LNCS}, vol. 1403, pp. 171--183. Springer,
  Heidelberg (May~/~Jun 1998)

\bibitem{Zajac2008PhD}
Zajac, P.: Discrete Logarithm Problem in Degree Six Finite Fields. Ph.D.
  thesis, Slovak University of Technology (2008),
  \url{http://www.kaivt.elf.stuba.sk/kaivt/Vyskum/XTRDL}

\end{thebibliography}
\def\noopsort#1{}\ifx\bibfrench\undefined\def\biling#1#2{#1}\else\def\biling#1#2{#2}\fi\def\Inpreparation{\biling{In
  preparation}{en
  pr{\'e}paration}}\def\Preprint{\biling{Preprint}{pr{\'e}version}}\def\Draft{\biling{Draft}{Manuscrit}}\def\Toappear{\biling{To
  appear}{\`A para\^\i tre}}\def\Inpress{\biling{In press}{Sous
  presse}}\def\Seealso{\biling{See also}{Voir
  {\'e}galement}}\def\Editor{\biling{Ed.}{R{\'e}d.}}

\end{document}